\documentclass[11pt
]{article} 
\usepackage[paperwidth=8.45in, paperheight=11in, top=1.25in, bottom=1.25in, left=1.00in, right=1.00in]{geometry} 
\usepackage{slashbox}
\usepackage{imakeidx} 
	\makeindex

\usepackage{amssymb,amsmath,amsthm}

\usepackage{mathtools}
\mathtoolsset{showonlyrefs}
\usepackage[dvipsnames]{xcolor} 
\usepackage{tikz}
\usepackage{tikz-3dplot}
\usepackage[latin1]{inputenc} 
\usepackage[T1]{fontenc}
\usepackage{lmodern}
\usepackage[justification=RaggedRight]{caption}
\usepackage[all,error]{onlyamsmath} 
\usepackage[style=alphabetic,natbib=true]{biblatex}
\bibliography{ICE_bib}
\makeatletter
\newsavebox\myboxA
\newsavebox\myboxB
\newlength\mylenA
\newcommand*\xoverline[2][0.75]{%
    \sbox{\myboxA}{$\m@th#2$}%
    \setbox\myboxB\null
    \ht\myboxB=\ht\myboxA%
    \dp\myboxB=\dp\myboxA%
    \wd\myboxB=#1\wd\myboxA
    \sbox\myboxB{$\m@th\overline{\copy\myboxB}$}
    \setlength\mylenA{\the\wd\myboxA}
    \addtolength\mylenA{-\the\wd\myboxB}%
    \ifdim\wd\myboxB<\wd\myboxA%
       \rlap{\hskip 0.5\mylenA\usebox\myboxB}{\usebox\myboxA}%
    \else
        \hskip -0.5\mylenA\rlap{\usebox\myboxA}{\hskip 0.5\mylenA\usebox\myboxB}%
    \fi}
\makeatother

\newcommand{\Xibar}{\xoverline[1.3]{\Xi}}

\newcommand{\E}{\mathbb{E}}

\newcommand{\R}{\mathbb{R}}

\newcommand{\conv}{\mathrm{conv}}
\newcommand{\convc}{\overline{\mathrm{conv}}}

\DeclareMathOperator*{\argmin}{arg \, min}
\DeclareMathOperator*{\argmax}{arg \, max}

\def\ind{{\mathchoice{1\mskip-4mu\mathrm l}{1\mskip-4mu\mathrm l}
{1\mskip-4.5mu\mathrm l}{1\mskip-5mu\mathrm l}}}

\newtheorem{theorem}{Theorem}[section]

\newtheorem{corollary}[theorem]{Corollary}
\newtheorem{proposition}[theorem]{Proposition}
\newtheorem{remark}[theorem]{Remark}

\newtheorem{example}[theorem]{Example}

\usepackage[normalem]{ulem}

\renewcommand{\le}{\leqslant}
\renewcommand{\leq}{\leqslant}
\renewcommand{\ge}{\geqslant}
\renewcommand{\geq}{\geqslant}

\title{\textbf{Examples and Counterexamples\\
 of Cost-efficiency in Incomplete Markets}}

\author{
Carole Bernard \thanks{%
C. Bernard, Grenoble Ecole de Management,  12~rue Pierre S\'{e}mard, 38000 Grenoble, France  (e-mail: 
\texttt{carole.bernard@grenoble-em.com}).}
\thanks{Vrije Universiteit Brussel, Faculty of Economics, Belgium.}
\and
Stephan Sturm \thanks{ %
S. Sturm, Worcester Polytechnic Institute, Department of Mathematical Sciences, 100 Institute Road, Worcester, MA 06109, USA
	(e-mail: \texttt{ssturm@wpi.edu})}
}

\begin{document}

\maketitle

\begin{abstract}
We present a number of examples and counterexamples to illustrate the results on cost-efficiency in an incomplete market obtained in  \cite{BS24}. These examples and counterexamples do not only illustrate the results obtained in  \cite{BS24}, but show the limitations of the results and the sharpness of the key assumptions. In particular, we make use of a simple 3-state model in which we are able to recover and illustrate all key results of the paper. This example also shows how our characterization of perfectly cost-efficient claims allows to solve an expected utility maximization problem in a simple incomplete market (trinomial model) and recover results from \cite[Chapter 3]{DS06}, there obtained using duality.
\end{abstract}

\vspace{5mm}
 
\begin{flushleft}
	 \textbf{Keywords:} Cost-efficiency, portfolio choice, law-invariant objective, utility maximization.\\
	 \textbf{Mathematics Subject Classification (2010):} 91G10, 60E15, 90B50.\\
	 \textbf{JEL classification:} C02, G11, D81, C61.
\end{flushleft}

\vspace{1cm}

{\textbf{Declarations of interest:} 
}   none.
\section{Introduction}

We present a number of examples and counterexamples to better understand the notion ofcost-efficiency in an incomplete market. In particular, we make use of a 3-state model, very simple though able to illustrate all results from \cite{BS24}. These results indeed also hold in a discrete market with  $n$ equiprobable states. We provide here a few elements how the various propositions from the paper can be rederived in this setting. \cite[Proposition 2.1]{BS24} can be proved as follows: In the discrete setting, let $X$ and $Y$ denote  two random variables that take values $\{x_1,x_2,...x_n\}$ and $\{y_1,y_2,...y_n\}$. As noted by \cite[Section 2.3]{he2016risk}, we have that $Y \preceq_{cx} X$ if and only if $Y$ is in the convex hull of $n!$ permutations of $\{x_1,x_2,...x_n\}$ (1.A.3 in \cite{marshall2011inequalities}), which can be interpreted as a description of all random variables with the cdf of $X$.

The assumption that $\xi$ is continuously distributed that we use throughout the paper, needs to be replaced by the assumption that $\xi$ takes $n$ distinct equiprobable values (in a discrete probability space with $n$ states), as in general the assumption that a distribution is continuous coincides in the finite case with the assumption that it is strictly increasing, both expressed by the fact that the random variable takes $n$ distinct values. In this case, the cost-efficient payoff solution to 
\begin{equation}
	\min_{\substack{Z \in \mathcal{X} \\ Z \in D(F)}}\E\bigl[ \xi Z\bigr]\label{CE}
\end{equation} in a discrete complete market setting can be written as
\begin{equation}
	Z^*:=F^{-1}(1-U_{\xi}),\label{form}
\end{equation}
where $U_{\xi}=\bigl(F_{\xi}(\xi)+F^-_{\xi}(\xi)\bigr)/2$ (where $F^-_{\xi}$ denotes the left limit of the cdf of $F_{\xi}$ of $\xi$). In particular, $1-U_{\xi}$ and $U_{\xi}$ have the same distribution, which we refer to as the uniform distribution. The theory for cost-efficiency in such a market was originally  studied by \cite{D88a,D88b} without providing an explicit representation, only stating that the cost-efficient payoff outcomes need to be oppositely ordered to the values taken by the state prices, which is exactly  what \eqref{form} ensures. 

In particular, the result \cite[Proposition 4.4]{BS24} holds in a market setting with $n$ equiprobable states. The proof follows similarly by taking $U$ uniform over the $n$ states, taking values $(i-0.5)/n$ for $i=1,...,n$ for instance. If $\xi^*$ does not take $n$ distinct values then one can use the randomization (e.g., using for instance a symmetric Bernoulli distribution when $\xi^*$ takes two equal  values). The proof of \cite[Proposition 4.8]{BS24} is illustrated in Appendix~\ref{proofprop4.6} in the context of the 3-state model introduced hereafter.

We recall that to study cost-efficiency as a distributional optimization problem in incomplete markets in full generality, we have to widen the set  of random variables to those being convex combinations of random variables with distribution $F$ (i.e., to $\conv(F)$). We then consider the following cost-efficiency problem
\begin{equation}\label{cvxminimaxA}
	\inf_{Z \in \convc(F)} c(Z)=\inf_{Z \in \convc(F)} \sup_{\xi \in \Xi} \E\bigl[\xi Z\bigr].
\end{equation}
We call a pair $(Z^*, \xi^*)$ a solution to problem \eqref{cvxminimaxA} if $(Z^*, \xi^*) \in \mathcal{Z} \times \mathcal{Y}$ with
\begin{equation}
	\mathcal{Z}  := \argmin_{\substack{Z \in \mathcal{X} \\ Z \in\convc(F)   }} \, \sup_{\substack{\xi \in \Xi}} \E\bigl[ \xi Z\bigr], \quad
	\mathcal{Y}  := \bigcup_{Z \in \mathcal{Z}}\argmax_{\substack{\zeta \in \Xibar}} \E\bigl[ \zeta Z\bigr].\label{7b}
\end{equation}
In a slight abuse of language we also often call  just $Z^*$ a solution to the problem \eqref{cvxminimaxA}.

\section{A 3-states Market Model \label{PCED3s}}
 
In the following we provide an example of a market setting in which the subset of "perfectly cost-efficient" distributions can be derived explicitly. We present here the major ideas and results, all detailed calculations for this example are provided in Appendix~\ref{AppD1bis}. We consider  a discrete, equiprobable probability space $\Omega = \{\omega_1, \omega_2, \omega_3\}$ with the family of pricing kernels $(\xi^u)_{u \in (0,1/3)}$, $\xi^u = \bigl(3u, 3-9u,6u\bigr)$ (corresponding to the asset dynamics $S_0= 2$, $S_T = (4,2,1)$). Here we use the vector shorthand $(a,b,c)$ to write $a\ind_{\{\omega_1\}} + b\ind_{\{\omega_2\}} + c\ind_{\{\omega_3\}},$ which allows us to use vector calculus. Let the distribution with cost to be minimized be given by the cdf
\[
F(t) = \frac{1}{3} \Bigl(\ind_{[x,\infty)}(t) + \ind_{[y,\infty)}(t) + \ind_{[z,\infty)}(t)\Bigr)\ \  \text{for some}\ \  x<y<z.
\]
The random variables corresponding to this distribution are the six permutations
\[
	(x,y,z); \, (x,z,y); \, (y,x,z); \, (y,z,x); \, (z,x,y); \, (z,y,x)
\]
while their convex closure is
\[
	\conv(F) = (a,b,6-a-b), \quad x \leq a,b \leq z, \, x+y \leq a+b \leq y+z.
\]

For the maximin problem we have by anti-comonotonicity with $\xi^u=\bigl(3u, 3-9u,6u\bigr)$ for $u\in(0,1)$ that three of the six possible r.v. with cdf $F$ can be discarded (as they do not satisfy that the first coordinate is larger than the last one). We find that the optimal value for the maximin problem is given by
\begin{align*}
	\max_{\xi^u} \min_{ Z \in D(F)} \E[\xi^u Z] 
	&= \left\{\begin{array}{cl}\frac{2x+2y+z}{5}&\text{if}\ 2x-3y+z\le0,\\
	\\
	\frac{2x+y+z}{4}&\text{if}\ 2x-3y+z\ge 0.\\
	\end{array}   \right.
\end{align*}
Furthermore, we have three cases for the optimizer that are reported in the first row of Table~\ref{tab:optimizer}. In all cases, there are no unique solutions. For the minimax problem we solve it by working out all cases explicitly. After some simplifications, we find
that
\begin{align*}
	\min_{ Z \in D(F)}\max_{\xi^u} \E[\xi^u Z] 
	&= \left\{\begin{array}{cl}y&\text{if}\ 2x-3y+z\le0,\\
	\\
	\frac{2x+z}{3}&\text{if}\ 2x-3y+z\ge 0.\\
	\end{array}   \right.
\end{align*}

We observe that the optimal solution to the minimax problem can have a different value than that of the maximin problem. Specifically, if $2x-3y+z\neq0$ then 
\[
	\max_{\xi^u} \min_{ Z \in D(F)} \E[\xi^u Z]<\min_{ Z \in D(F)}\max_{\xi^u} \E[\xi^u Z].
\]
Notice from the fourth row of Table ~\ref{tab:optimizer} that the $Z$ component of the optimal solution can also be part of an optimal solution to the maximin problem, but the corresponding pricing kernels are different.

We then consider the convexified version of the minimax problem. We find that
\begin{align*}
	\min_{\conv(F)} \max_{\xi^u} \E[\xi^u Z] & = \left\{ \begin{array}{cl} \frac{2x+2y+z}{5}&\text{if}\ z-3y+2x\le 0,\\
	\frac{2x+y+z}{4}& \text{if}\ z-3y+2x\ge 0.\end{array}\right.
\end{align*}
There are two cases. If $z-3y+2x\le 0,$ the set of optimizers is 
\[
\mathcal{Z} \times \mathcal{Y} = \biggl(\Bigl(z,\frac{2x+2y+z}{5},\frac{3x+3y-z}{5}\Bigr), (3u,3-9u,6u)\biggr)\quad \text{for}\quad u\in[0,1/3].
\] 
If $z-3y+2x\ge 0,$ the set of optimizers is 
\[
\mathcal{Z} \times \mathcal{Y} = \biggl(\Bigl(\frac{3z+3y-2x}{4},\frac{2x+y+z}{4},x\Bigr),(3u,3-9u,6u)\biggr)\quad \text{for} \quad u\in[0,1/3].
\]
We note that in both cases, the random variable $Z^*$ is unique, but the pricing kernel is not. In fact all pricing kernels are optimal, but the optimal random variable does not have the distribution $F$ in general. $Z^*$ is thus cost-efficient but not always perfectly cost-efficient.

Using the expression of the optima, we can easily identify the set of perfectly cost-efficient payoffs. To be perfectly cost-efficient, they must have the distribution $F$ and thus take the 3 values $x$, $y$ and $z$. If $z-3y+2x\le 0,$ observe that since $x<y<z$, we have 
\[
	z>\frac{2x+2y+z}{5}>\frac{3x+3y-z}{5}.
\]
To find the set of all perfectly cost-efficient distributions, we then solve $(2x+2y+z)/5=y$ and $(3x+3y-z)/5=x$. We find that this corresponds to 

\begin{equation}\label{perfCE}
	F(t) = \frac{1}{3} \Bigl(\ind_{[x,\infty)}(t) + \ind_{[y,\infty)}(t) + \ind_{[3y-2x,\infty)}(t)\Bigr)\ \  \text{for}\ \  x<y,
\end{equation}
which ensures that the optimizer has the distribution $F$ if and only if $3y-2x=z$. If $z-3y+2x\ge 0,$ observe that for $x<y<z$, we also have
\[
	\frac{3z+3y-2x}{4}>\frac{2x+y+z}{4}>x
\]
To find the set of all perfectly cost-efficient distributions, we then solve $(3z+3y-2x)/4=z$ and $(2x+y+z)/4=y$. We find that it corresponds to \eqref{perfCE}. Thus in both cases, the set of perfectly cost-efficient distributions reduces to \eqref{perfCE}.

\begin{remark}
	Note that this set of (perfectly) cost-efficient distributions is in bijection with $\R^2$ while the original set of all possible distributions over 3-states was in bijection with $\R^3$. This does not come as surprise if one takes the point of view of hedging: All attainable portfolios can be written as $\theta S_T + (x_0 - S_0\theta)$, characterized by the two parameter family $(x_0, \theta)$ describing initial capital and number of shares invested. We see that, for arbitrary $x_0$, the cost-efficient portfolios correspond exactly to those with a long stock position,  $\theta \geq 0$, illustrating the equivalence $(ii) \Leftrightarrow (v)$ in  \cite[Theorem 5.3]{BS24}.
\end{remark}

Finally, for the convexified version of the maximin problem $\max_{\xi^u} \min_{\conv(F)} \E[\xi^u Z]$, there are also  three cases depending on the sign of $2x-3y+z$. A comparative analysis of the optima and optimizers of the four problems is displayed in Tables~\ref{tab:optimum} and~\ref{tab:optimizer}.

This example illustrates  the equality and inequality in \eqref{CEFF} in  \cite[Proposition~4.2]{BS24},
\begin{equation}\label{CEFF}
		\sup_{\xi \in \Xi} \, \inf_{\substack{Z \in \mathcal{X} \\ Z \in D(F)}} \E\bigl[ \xi Z\bigr] =  \sup_{\xi \in \Xi} \, \inf_{Z \in \convc(F)} \E\bigl[ \xi Z\bigr] =  \inf_{Z \in \convc(F)} \, \sup_{\xi \in \Xi} \E\bigl[ \xi Z\bigr] \leq \inf_{\substack{Z \in \mathcal{X} \\ Z \in D(F)}} \, \sup_{\xi \in \Xi}\E\bigl[ \xi Z\bigr],
	\end{equation}
	and sheds some light on the subtle relationship between the minimax problem (in \eqref{dist0}) and the maximin problem (in  \eqref{dist-problem}) (as well as their convexified versions) that we recall here in their general form for convenience: The minimax problem is
\begin{equation} \label{dist0}
	\inf_{\substack{Z \in \mathcal{X} \\ Z \in D(F)}} c(Z)=
	\inf_{\substack{Z \in \mathcal{X} \\ Z \in D(F)}} \sup_{\xi \in \Xi} \E\bigl[ \xi Z\bigr],
\end{equation}
where $c(Z)$ denotes the superhedging price \eqref{SHP}, i.e., we associate to every claim $X \in \mathcal{X}$ its cost $c(X)$, the amount of money needed to superreplicate it in all possible states of the world, that satisfies (see, e.g., \cite[Section 3]{K96})
\begin{equation}\label{SHP}
    c(X) = \sup_{\xi \in \Xi} \E[\xi X].
\end{equation}
The maximin problem writes as
\begin{equation}\label{dist-problem}
	\sup_{\xi \in \Xi} \, \inf_{\substack{Z \in \mathcal{X} \\ Z \in D(F)}} \E\bigl[ \xi Z\bigr].
\end{equation}

Overall we realize that here as proved in the continuous setting, the values of the maximin problem, convexified maximin problem and convexified minimax problem agree, the value of the minimax problem might be strictly larger (see Table~\ref{tab:optimum} for an overview of the optimal values). Furthermore, the solution of the problems are not unique here, similarly to what was observed in the continuous model. Even in the perfectly cost-efficient case each optimization problem has idiosyncratic solutions beyond the shared family $\bigl((z,y,x),\xi^u\bigr), \, u \in [1/5,1/4]$.

\begin{table}[!htbp]
	\centering
	{\small
		\begin{tabular}{r|c|c|c}
			&$2x-3y+z>0$&$2x-3y+z=0$&$2x-3y+z<0$\\\hline\hline
			&&&\\
			$\displaystyle \max_{\xi^u} \min_{ Z \in D(F)} \E[\xi^u Z]=$&
			$\frac{2x+y+z}{4}$ &
			$\frac{2x+y+z}{4}= y = \frac{2x+2y+z}{5}$&
			$\frac{2x+2y+z}{5}$\\
			&&&\\\hline
			&&&\\
			$\displaystyle \max_{\xi^u} \min_{ Z \in \conv{(F)}}  \E[\xi^u Z]=$&
			$\frac{2x+y+z}{4}$ &
			$\frac{2x+y+z}{4}= y = \frac{2x+2y+z}{5}$&
			$\frac{2x+2y+z}{5}$\\
			&&&\\\hline
			&&&\\
			$\displaystyle \min_{ Z \in \conv{(F)}} \max_{\xi^u} \E[\xi^u Z]=$&
			$\frac{2x+y+z}{4}$ &
			$\frac{2x+y+z}{4}= y = \frac{2x+2y+z}{5}$&
			$\frac{2x+2y+z}{5}$\\
			&&&\\\hline
			&&&\\
			$\displaystyle \min_{ Z \in D(F)} \max_{\xi^u} \E[\xi^u Z]=$&$\frac{2x+z}{3}$&
			$\frac{2x+z}{3}=y$&
			$y$\\
			&&&\\
		\end{tabular} }
	\caption{Overview of the optimal values of the different optimization problems. Note that the first three problems have always the same solutions, whereas the last (minimax) problem has this solution only in the perfectly cost-efficient case (middle column), being larger otherwise.}
\label{tab:optimum}
\end{table}

\begin{remark}\label{attainable}
	In the case when the distribution $(x,y,z) \in D(F), \, x < y < z$ is perfectly cost-efficient, i.e., $z= 3y-2x$, the only possible payoff that is attainable that writes as $F^{-1}(1-U_{\xi^u})$ or its randomization (where $U_{\xi^u}$ is defined in \eqref{form}) is $(z,y,x)$ for $u \in[1/5,1/4]$ (as  it would imply that for the pricing kernel $\xi^u$ we have $3u \leq 3-9u \leq 6u$, which is equivalent to $u\in [1/5,1/4]$). Details can be found in Appendix~\ref{proofattain}.
\end{remark}

\begin{table}[!tbp]
	\centering
	\scalebox{0.65}{
	{\small
		\begin{tabular}{r|c|c|c}
			&$2x-3y+z>0$&$2x-3y+z=0$&$2x-3y+z<0$\\\hline\hline
			&&&\\
			$\displaystyle \max_{\xi^u} \min_{ Z \in D(F)}$&
			\begin{tabular}{c}
				$\bigl((z,y,x), \xi^\frac{1}{4}\bigr)$\\
				\\
				or $\bigl((y,z,x), \xi^\frac{1}{4}\bigr)
				$\end{tabular}
			&
			\begin{tabular}{c}
				$
				\bigl((z,x,y), \xi^\frac{1}{5}\bigr)$\\
				\\
				or $\bigl((y,z,x), \xi^\frac{1}{4}\bigr)$\\
				\\
				or $\left\{
				\bigl((z,y,x), \xi^u\bigr)\right\}_{u\in [1/5,1/4]}
				$\end{tabular}
			&
			\begin{tabular}{c}
				$\bigl((z,x,y), \xi^\frac{1}{5}\bigr)$\\
				\\
				or $\bigl((z,y,x), \xi^\frac{1}{5}\bigr)
				$\end{tabular}\\
			&&&\\\hline
			&&&\\
			$\displaystyle \max_{\xi^u} \min_{ Z \in \conv{(F)}} $&
			$\left\{(z-t,y+t,x), \xi^\frac{1}{4}\right\}_{t \in [0, z-y]}$ &
			$\begin{array}{c} \left\{\bigl((z,y,x), \xi^u\bigr)\right\}_{u \in (\frac{1}{5},\frac{1}{4})} \\ \text{or } \Bigl\{\bigl((z-t,y+t,x), \xi^\frac{1}{4}\bigr)\Bigr\}_{t \in [0,z-y]}\\ \text{or } \Bigl\{\bigl((z,y-t,x+t), \xi^\frac{1}{5} \bigr)\Bigr\}_{t \in [0,y-x]}
			\end{array}$
			&
			$\left\{\left((z,y-t,x+t), \xi^\frac{1}{5}\right)\right\}_{t \in [0, y-x]}$\\
			&&&\\\hline
			&&&\\
			$\displaystyle \min_{ Z \in \conv{(F)}} \max_{\xi^u} $&
			$\left\{\Bigl(\bigl(\frac{-2x+3y+3z}{4},\frac{2x+y+z}{4},x\bigr),\xi^{u}\Bigr)\right\}_{u\in[0,\frac{1}{3}]}$ &
			$\left\{\Bigl(\bigl(z,y,x\bigr),\xi^{u}\Bigr)\right\}_{u\in[0,\frac{1}{3}]}$&
			$\left\{\Bigl(\bigl(z,\frac{2x+2y+z}{5},\frac{3x+3y-z}{5}\bigr),\xi^{u}\Bigr)\right\}_{u\in[0,\frac{1}{3}]}$\\
			&&&\\\hline
			&&&\\
			$\displaystyle \min_{ Z \in D(F)} \max_{\xi^u} $&
			$\Bigl((z,y,x),\xi^{\frac{1}{3}}\Bigr)$
			&
			$\Bigl\{\Bigl((z,y,x),\xi^{u}\Bigr)\Bigr\}_{u\in[0,1/3]}$
			&
			$
			\begin{array}{cc}
			\Bigl((z,y,x),\xi^{0}\Bigr)&\hbox{if}\ x-3y+2z>0\\
			\left\{\begin{array}{l}\Bigl\{\Bigl((x,y,z),\xi^{u}\Bigr)\Bigr\}_{u\in[0,1/3]} \\ \Bigl((z,y,x),\xi^{0}\Bigr)\end{array}\right\}&\hbox{if}\ x-3y+2z=0\\
			\left\{\begin{array}{l}\Bigl((x,y,z),\xi^{0}\Bigr) \\ \Bigl((z,y,x),\xi^{0}\Bigr)\end{array}\right\}&\hbox{if}\ x-3y+2z<0\\
			\end{array}$\\
			&&&\\
	\end{tabular} }}
	\caption{Overview of the optimizers $(Z^*,\xi^*)$ of the different optimization problems. Note that in the perfectly cost-efficient case (middle column), all problems share the optimizers $\bigl\{\bigl((z,y,x), \xi^u\bigr)\bigr\}_{u \in [\frac{1}{5},\frac{1}{4}]}$.}
	\label{tab:optimizer}
\end{table}

Let us give in the context of this example also an example of a payoff obtained from a preference optimization that fails to be perfectly cost-efficient, due to the fact that the preference functional lacks the diversification-loving property (while being law-invariant, increasing and upper semicontinuous):

\begin{example}
	Consider the expected utility maximization problem with utility function
	\[
		U(x) = \left\{ \begin{array}{ll} x^2 & \text {if } x \geq 0, \\-\infty & \text{if } x<0, \end{array} \right.
	\]
	and initial capital $x_0=1$ in the 3-states model. Investing in $\theta$ shares of the risky asset yields
	\[
	\sup_{X \in \mathcal{X}(1)} \E\bigl[U(X)\bigr] = \max_{\theta \in [-1,\frac{1}{2}]} \frac{1}{3}U\Bigl((4,2,1)\theta +(1-2\theta)\Bigr) = \frac{14}{15}
	\]
reached at the payoff $\bigl(\frac{3}{5},1,\frac{6}{5}\bigr)$ (for $\theta = -\frac{1}{5}$). One easily checks that this payoff is not perfectly cost-efficient (and thus not cost-efficient at all, cf.   \cite[Remark~3.6]{BS24}) as the cost minimization problem for the cdf
\[
F(t) = \frac{1}{3}\Bigl(\ind_{[\frac{3}{5}, \infty)}(t) + \frac{1}{3}\ind_{[1, \infty)}(t) + \frac{1}{3}\ind_{[\frac{6}{5}, \infty)}(t)\Bigr)
\]
yields, using Table~\ref{tab:optimizer} the optimizer $\bigl(\frac{6}{5},\frac{22}{25},\frac{18}{25}\bigr) \prec_{cx} \bigl(\frac{3}{5},1,\frac{6}{5}\bigr)$.
\end{example}

\section{Expected Utility Maximization in the 3-states Model \label{ExampEU}}

Let $u$ be an increasing concave utility, then the maximum expected utility maximization problem
\begin{equation*}
	\sup_{X \in \mathcal{D}_{x_0}} \E[u(X)]
\end{equation*}
has a unique solution $(3x_0-2x^*,x_0,x^*)$, where $x^*$ solves 
\begin{equation}\label{AA}
	\max_{x\le x_0} \Bigl(u(3x_0-2x)+u(x)\Bigr).
\end{equation}
Indeed, we have that $ \E[u(X)]=u(x_1)/3+u(x_2)/3+u(x_3)/3.$ Furthermore, as the optimum is perfectly cost-efficient, it is of the form $x_1=3y-2x$, $x_2=y$ and $x_3=x$ (\cite[Theorem~5.3]{BS24} and Remark~\ref{attainable}). The constraint on the budget $c(X)\le x_0$ simplifies to $y={x_0}.$ Then, the maximum expected utility problem is simply \eqref{AA}. Furthermore, when $u$ is differentiable, the first-order conditions can also be written as
\[
	u^\prime(x^*)=2u^\prime\bigl(3x_0-2x^*\bigr).
\]
We also easily see that in particular $x^*<x_0$.

Our approach allows to recover the results obtained using duality by \cite{DS06} in the example of a  trinomial model. Specifically, 
in Example 3.3.4 of \cite[Chapter 3]{DS06}, we further assume that $m=1/3$, i.e., the three states labeled $\{n,g,b\}$ in the paper are then equiprobable. With their notation,
\[
	S_1=\left\{\begin{array}{l}
	S_0(1+\tilde u),\\
	S_0,\\
	S_0(1+\tilde d),\\
	\end{array}\right.
\]
where we choose $S_0=2$, $\tilde u=1$ and $\tilde d=-1/2$ to recover our 3-states example. We have that $1+\tilde u>1>1+\tilde d>0$ and $\tilde u>-\tilde d$. In this case, \cite{DS06} show that the optimal payoff maximizing the expected utility is given as 
\[
	X^*=\left\{\begin{array}{cl}
	x_0+\hat h \tilde u &\hbox{if}\ S_1=S_0(1+\tilde u),\\
	x_0 &\hbox{if}\ S_1=S_0,\\
	x_0+\hat h \tilde d &\hbox{if}\ S_1=S_0(1+\tilde d)\\
	\end{array}\right.
\]
and compute its expression for three choices of concave utility function.

Denote by $q:=-\tilde d/(\tilde u-\tilde d)=1/3.$ In Example 3.3.4 of \cite[Chapter 3]{DS06}, when $u(x)=\ln(x)$,
\[
	\hat h=x\frac{\tilde d+\tilde u}{-2\tilde d\tilde u},\quad 
	X^*=\left\{\begin{array}{cl}
	\frac{x_0}{2q}=\frac{3x_0}{2} &\hbox{if}\ S_1=S_0(1+\tilde u),\\
	x_0 &\hbox{if}\ S_1=S_0,\\
	\frac{x_0}{2(1-q)}=\frac{3x_0}{4} &\hbox{if}\ S_1=S_0(1+\tilde d),\\
	\end{array}\right.
\]
which is consistent with our result. It writes as $(3x_0-2x^*,x_0,x^*)$ where $x^*=3x_0/4$ solves \eqref{AA} when  $u(x)=\ln(x)$.

When $u(x)=-\exp(-x)$,
\[
	\hat h=\frac{1}{\tilde u-\tilde d}\ln\Bigl(\frac{\tilde u}{-\tilde d}\Bigr)=\frac{2}{3}\ln(2),\quad X^*=\left\{\begin{array}{cl}
	x_0+\hat h \tilde u &\hbox{if}\ S_1=S_0(1+\tilde u),\\
	x_0 &\hbox{if}\ S_1=S_0,\\
	x_0+\hat h \tilde d &\hbox{if}\ S_1=S_0(1+\tilde d),\\
	\end{array}\right.
\]
which is consistent with our result. It writes as $(3x_0-2x^*,x_0,x^*)$ where $x^*=x_0-\ln(2)/3$ solves \eqref{AA} when  $u(x)=-\exp(-x)$.

When $u(x)=x^\alpha/\alpha$, with $\alpha<1$ and $\alpha\neq0$. Define $\beta=\alpha/(\alpha-1)$ and $c_v=1/2\bigl((2q)^\beta+(2(1-q))^\beta\bigr)=1/3^\beta\bigl(2^{\beta-1}+2^{2\beta-1}\bigr)$, then
\[
	\hat h=\frac{x_{0}}{\tilde u}\biggl(\frac{(2q)^{\beta-1}}{c_v}-1\biggr)={x_{0}}\Bigl(\frac{3}{1+2^\beta}-1\Bigr),\quad X^*=\left\{\begin{array}{cl}
	x_0+\hat h \tilde u &\hbox{if}\ S_1=S_0(1+\tilde u),\\
	x_0 &\hbox{if}\ S_1=S_0,\\
	x_0+\hat h \tilde d &\hbox{if}\ S_1=S_0(1+\tilde d).\\
	\end{array}\right.
\]
We then find that
\[
	X^*=\biggl(\ \frac{3x_{0}}{1+2^\beta},\ x_{0},\ 3x_{0}\frac{2^{\beta-1}}{1+2^\beta}\ \biggr),
\]
which  is also consistent with our result as it writes $(3x_0-2x^*,x_0,x^*)$ where  the optimal solution to \eqref{AA} for the power utility function is given by
\[
	x^*=3x_{0}\frac{2^{\beta-1}}{1+2^\beta}.
\]

\section{Difference Between Minimax Problem \eqref{dist0} and Maximin Problem \eqref{dist-problem}}\label{sec:difference}

The last inequality in \cite[Proposition~4.2]{BS24} (recalled in  \eqref{CEFF}) can be strict as stated by the following proposition.

\begin{proposition}\label{p4.6}
  Minimax Problem \eqref{dist0} and Maximin Problem \eqref{dist-problem} do not share the same solution in general.
\end{proposition}

\begin{proof}
	Examples~\ref{coexad} and~\ref{coexa} hereafter provide  counterexamples first in a discrete market and then in a continuous market and thus prove Proposition \ref{p4.6}. Note that  \cite[Theorem 5.3]{BS24} shows that this can only happen when the target distribution is not perfectly cost-efficient.   Therefore, Example~\ref{coexad} is  an example in a discrete setting in which the target distribution  is not perfectly cost-efficient. In Example~\ref{coexa}, we show that the distribution of the underlying stock price is not perfectly cost-efficient.
\end{proof}
	
\begin{example}\label{coexad}
	Using the example in Section~\ref{PCED3s}, for the cdf $F(t) = \bigl(\ind_{[1,\infty)}(t) + \ind_{[2,\infty)}(t) + \ind_{[3,\infty)}(t)\bigr)/3$ (i.e., $x=1$, $y=2$ and $z=3$).
	\begin{align*}
		\max_{\xi^u\in\Xibar} \min_{ Z \in D(F)} \E[\xi^u Z] & = \frac{9}{5}.
	\end{align*}
	The optimal solution is reached at $(Z_1^*, \xi_1^*) = \bigl((3,1,2), (\frac{3}{5}, \frac{6}{5}, \frac{6}{5})\bigr)$ and $(Z_2^*, \xi_2^*) =\bigl((3,2,1),\bigr.$ $\bigl. (\frac{3}{5}, \frac{6}{5}, \frac{6}{5})\bigr)$. For the minimax problem we get by doing all cases explicitly 
	\begin{align*}
		\min_{ Z \in D(F)} \max_{\xi^u} \E[\xi^u Z] &  = 2,
	\end{align*}
	and that the optimal solution is reached at $(\tilde{Z}, \tilde{\xi}) = \bigl((3,2,1), (0,3,0)\bigr)$. Obviously this gives a larger value than the maximin problem as $2 > 1.8$. Notice that the $Z$ component of the optimal solution to the minimax problem is also a part of an optimal solution to the maximin problem, but the corresponding pricing kernel is different as $\tilde{\xi}\in\Xibar$ is obtained for $u=0$ (which is not anymore a pricing kernel but a limit of pricing kernels).
\end{example}	

\begin{example}\label{coexa}
	Consider a Black--Scholes model with one risky asset that has constant drift $\mu$ and uncertain volatility $\sigma(X)$ that depends on the regime $X$. This volatility can be equal to $\sigma_H$ (when $X$ indicates a high volatility market) with probability $p\in(0,1)$, and $\sigma_L$ (when $X$ indicates a low volatility market) with probability $1-p$. We further assume that $X$ is independent of the Brownian motion $W$, and that there is a risk-free asset that pays no interest,
	\begin{align*}
		B_t & =1\quad\forall t \geq 0,\\
		dS_t & =\mu S_t dt +\sigma(X) S_t dW_t,
	\end{align*}
	where $S_0=s.$ The market price of risk $\theta_t$ at time $t$ is defined through
	\[
		\theta\sigma=\mu.
	\]
	As in 
	\cite[Lemma 4.4]{BS14}, we have that 
	\[
		\xi_T^q=\frac{q}{p}\mathcal{E}\biggl(-\int_0^{\cdot}\theta_H dW_t\biggr)_T \ind_{\{\sigma = \sigma_H\}}+\frac{1-q}{1-p}\mathcal{E}\biggl(-\int_0^{\cdot}\theta_L dW_t\biggr)_T\ind_{\{\sigma = \sigma_L\}}
	\]
	where $\theta_L:=\mu/\sigma_L$ and $\theta_H=\mu/\sigma_H$. The stock price costs $S_0$ at time $0$ and achieves a payoff $S_T$ at time $T$. Both the stock price $S_{T}$ and the pricing kernel $\xi_{T}^q$ have a mixture distribution, for which the cdfs and quantiles can be computed explicitly. We are able to conclude that the cost of the cheapest attainable payoff to achieve the distribution of final wealth $F_{S}$ is strictly larger than the left side of the inequality in \cite[Proposition~4.2]{BS24}. Thus solutions to problems \eqref{dist0} and \eqref{dist-problem} are not the same. Detailed calculations are given in Appendix~\ref{Appcoex}.
\end{example}

\section{Cost and Convex Order}

It is immediate from \cite[Proposition~4.2]{BS24} that in a complete market with continuous pricing kernel $\xi$, if a payoff $X$ is smaller in convex order than a cost-efficient payoff $X^*$, then it is more expensive, i.e., $c(X)\ge c(X^*)$. Because of the completeness of the market,  the inequality at the end of \eqref{CEFF} becomes an equality. If $F$ has finite mean, \eqref{CEFF}  is equivalent to
	\begin{equation}\label{CEFF2}
		\sup_{\xi \in \Xi} \, \inf_{\substack{Z \in \mathcal{X} \\ Z \in D(F)}} \E\bigl[ \xi Z\bigr] =  \sup_{\xi \in \Xi} \, \inf_{\substack{Z \in \mathcal{X} \\ F_{Z}\preceq_{cx} F}} \E\bigl[ \xi Z\bigr] =  \inf_{\substack{Z \in \mathcal{X} \\ F_{Z}\preceq_{cx} F}} \, c(Z)\leq \inf_{\substack{Z \in \mathcal{X} \\ Z \in D(F)}} \, c(Z).
	\end{equation}
 Let $X^*$ be cost-efficient with cdf $F$ (so $X^*=F^{-1}\bigl(1- \hat{F}_{\xi}(\xi,U)\bigr)$). Thus 
 \eqref{CEFF2} becomes  
\[
	\inf_{\substack{Z \in \mathcal{X} \\ Z\preceq_{cx} X^*}} \E\bigl[ \xi Z\bigr]= \inf_{\substack{Z \in \mathcal{X} \\ F_{Z}\preceq_{cx} F}} \,  \E\bigl[ \xi Z\bigr]=\inf_{\substack{Z \in \mathcal{X} \\ Z \in D(F)}} \E\bigl[ \xi Z\bigr]=\E\bigl[ \xi F^{-1}\bigl(1- \hat{F}_{\xi}(\xi,U)\bigr)\bigr],
\]
and any  payoff $Z$ convex smaller than $X^*$ has obviously  its price larger than the infimum price. This observation slightly extends the result in \cite[Lemma 2]{bernard2017optimal}, which states  that in a complete market,  the cost-efficient payoff for a claim with a distribution smaller in convex order is more expensive than the cost-efficient payoff corresponding to a claim with distribution larger in convex order. In an incomplete market, the same is true and can be proved straightforwardly. 

\begin{corollary}\label{convex_order}
	Consider a market with $F^1$, $F^2$ be two distributions with support bounded from below, with identical finite expectations such that $F^1 \preceq_{cx} F^2$. Then the superhedging cost of $F^1$ is larger than or equal to that of $F^2$.
\end{corollary}

\begin{proof}
	Indeed, by transitivity of the convex order, for random variables $Z$ we have that $F_Z \preceq_{cx} F^{\prime} \preceq_{cx} F^{\prime\prime}$ implies $\{Z \in \mathcal{X} \, : \, F_Z \preceq_{cx} F^1\} \subseteq \{Z \in \mathcal{X} \, : \, F_Z \preceq_{cx} F^2\}$ and therefore
	\begin{equation}\label{pricecomp}
		\inf_{\substack{Z \in \mathcal{X} \\ F_{Z}\preceq_{cx} F^2}} \, c(Z) \leq \inf_{\substack{Z \in \mathcal{X} \\ F_{Z}\preceq_{cx} F^1}} \, c(Z).
	\end{equation}
\end{proof}

\begin{remark}
	Note that while the order of the superhedging prices follows quite naturally from the setting, it is (contrary to complete markets) no longer true that the actual optimizing payoffs are ordered in convex order. While, denoting the cost-minimizers of $F^1$ and $F^2$ with $Z^{*,1}$ and $Z^{*,2}$ respectively obviously $Z^{*,2}$ cannot be larger than $Z^{*,1}$ in convex order, it does not have to be smaller (or equal). To show this, we rely another time on the 3-states model. Let
	\begin{align*}
		F^1(t) &= \frac{1}{3} \Bigl(\ind_{[1,\infty)}(t) + \ind_{[2,\infty)}(t) + \ind_{[4,\infty)}(t)\Bigr),\\
		F^2(t) &= \frac{1}{3} \Bigl(\ind_{[\frac{3}{2},\infty)}(t) + \ind_{[2,\infty)}(t) + \ind_{[\frac{7}{2},\infty)}(t)\Bigr).
	\end{align*}
Of course we have $F^1 \preceq_{cx} F^2$. $F^1$ is actually perfectly cost-efficient with optimizer $Z^{1,*} = (4,2,1)$ at cost $c\bigl(Z^{1,*}\bigr)=2$. However, $F^2$ is not, and the optimizer can derived by Table~\ref{tab:optimizer}, it is $Z^{2,*} = \bigl(\frac{27}{8},\frac{17}{8},\frac{3}{2}\bigr)$  at cost $c\bigl(Z^{2,*}\bigr)=1.9$. One easily checks that the two optimizers do not compare in convex order, as $\E\bigl[(2-Z^{1,*})^+\bigr] = 1/3 < 1/2 = \E\bigl[(2-Z^{2,*})^+\bigr]$ while $\E\bigl[(Z^{1,*}-2)^+\bigr] = 2/3 > 1/2 = \E\bigl[(Z^{2,*}-1)^+\bigr]$.
\end{remark}

To illustrate Corollary~\ref{convex_order},  we perform a comparison of the cost of the cost-efficient payoff with target LogNormal distribution and with target Normal distribution in an incomplete market. We use \cite[Proposition~4.4]{BS24} that solves problem \eqref{dist-problem}:
\begin{equation*}
    \sup_{\xi \in \Xi} \, \inf_{\substack{Z \in \mathcal{X} \\ Z \in D(F)}} \E\bigl[ \xi Z\bigr].
 \end{equation*}
Consider the stochastic volatility model used in Example~\ref{coexa}. Then the condition of \cite[Remark~4.5]{BS24} is satisfied and the optimal solution can be written of the form
\begin{equation}\label{solutionRepeated}
	Z^* = F^{-1}\bigl(1-F_{\xi^*} (\xi^*)\bigr).
\end{equation}
We then use a range of probability distributions for which we compute numerically the cheapest superhedging price.  We use the mean and variance
\[
m=S_0e^{\mu T}, \qquad V_{\sigma,p}:=\Bigl(pe^{\sigma_{H}^2T}+(1-p)e^{\sigma_{L}^2T}-1\Bigr)S_{0}^2e^{2\mu T},
\]
of $S_T$ to compute the corresponding parameters of a LogNormal distribution and a Normal distribution with equal mean and equal variance.

For the LogNormal distribution LN$(M,s^2)$, we have that 
\[
s=\sqrt{\ln\bigl(1+\frac{V_{\sigma,p}}{S_{0}^2e^{2\mu T}}\bigr)}, \qquad M=\ln(S_{0}e^{\mu T})-\frac{s^2}{2},
\]
and for the Normal distribution N($m,\sigma^2)$,  $m=S_{0}e^{\mu T}$ and $\sigma^2=V_{\sigma,p}.$
We then represent the superhedging price for each level of variance for the Normal and LogNormal distribution in Figure~\ref{F2}. We find that the distributions larger in convex order are also cheaper.

\begin{figure}[!h]
	\centering
	\includegraphics[width=15cm,height=7cm]{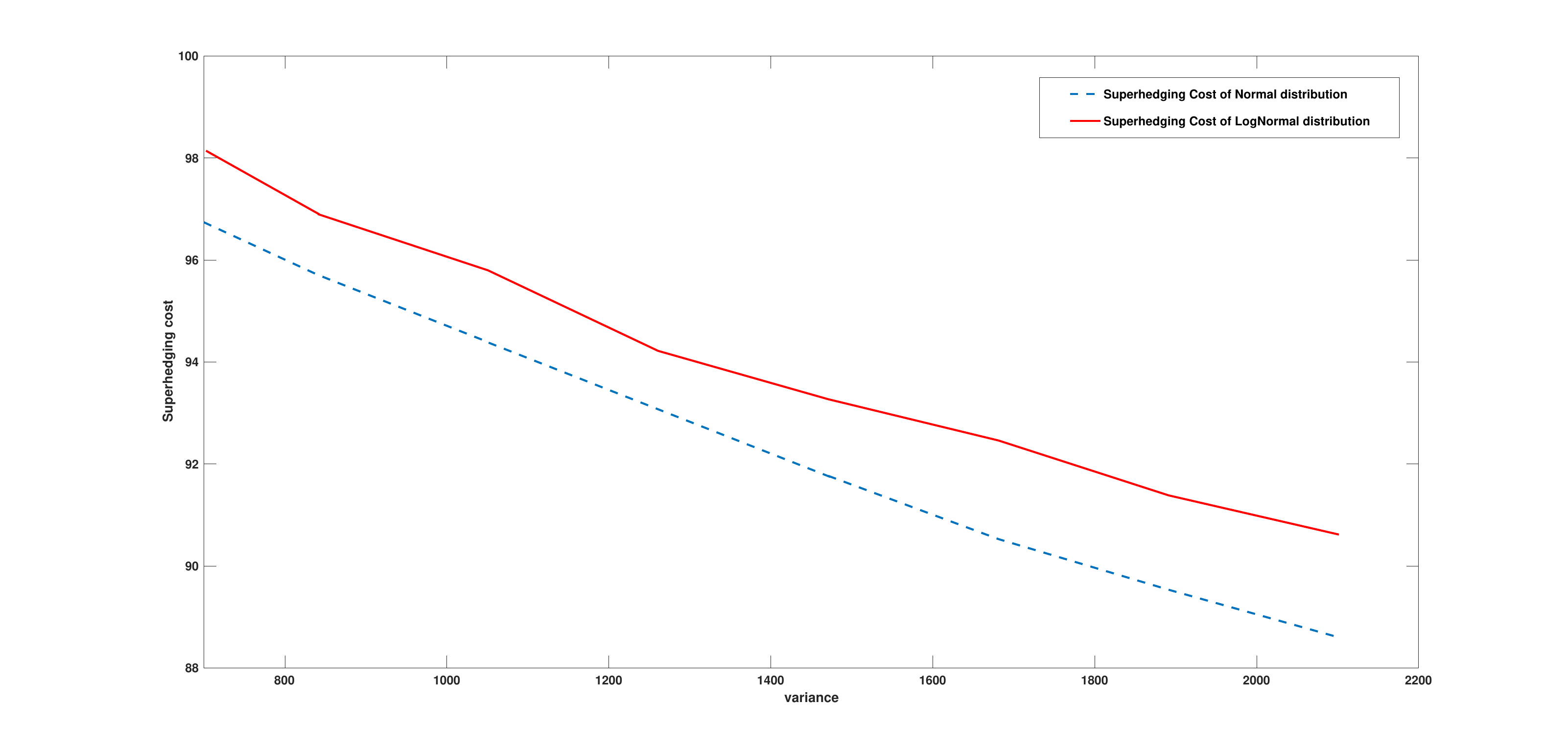} 
	\caption{Illustration of the fact that distributional superhedging costs decrease with convex order: Superhedging costs of normal distributions with fixed means as functions of the variance in the stochastic volatility model of Example~\ref{coexa}.}\label{F2}
\end{figure}
\newpage

\section*{Acknowledgments}
 
We would like to thank an associate editor for very detailed and constructive comments that helped us revise our paper. All remaining errors are ours. We thank Alfred M\"uller, Moris Strub and Ruodu Wang for their feedback on an earlier version of this paper. Also thanks to the Institute for Pure \& Applied Mathematics (IPAM), where at the Long Program on Broad Perspectives and New Directions in Financial Mathematics the seeds for this collaboration were laid. C. Bernard acknowledges funding from FWO G093024N at VUB.

\printbibliography

%

\newpage
\appendix

\begin{center}
{\LARGE \textbf{Appendix}}
\end{center}

\section{Detailed Calculation for a General 3-states Example \label{AppD1bis} }

Consider a discrete, equiprobable probability space $\Omega = \{\omega_1, \omega_2, \omega_3\}$ and the family of pricing kernels $(\xi^u)_{u \in (0,1/3)}$, $\xi^u = \bigl(3u, 3-9u,6u\bigr)$ (corresponding to the asset dynamics $S_0= 2$, $S_T = (4,2,1)$). Here we use the vector shorthand $(a,b,c)$ to write $a\ind_{\{\omega_1\}} + b\ind_{\{\omega_2\}} + c\ind_{\{\omega_3\}}$, which allows us to use vector calculus. Let the distribution with cost to be minimized given by the cdf $F(t) =  \bigl(\ind_{[x,\infty)}(t) + \ind_{[y,\infty)}(t) + \ind_{[z,\infty)}(t)\bigr)/3$ for some $x<y<z$. The random variables corresponding to this distribution are the six permutations
\[
(x,y,z); \, (x,z,y); \, (y,x,z); \, (y,z,x); \, (z,x,y); \, (z,y,x)
\]
while their convex closure is
\[
\conv(F) = (a,b,6-a-b), \quad x \leq a,b \leq z, \, x+y \leq a+b \leq y+z.
\]

\subsection{Maximin Problem \eqref{dist-problem}}

For the maximin problem \eqref{dist-problem}, we have by anti-comonotonicity with $\xi^u=\bigl(3u, 3-9u,6u\bigr)$ for $u\in(0,1)$ that three of the six possible r.v. with cdf $F$ can be discarded (as they do not satisfy that the first coordinate is larger than the last one),
\begin{align*}
	\max_{\xi^u} \min_{ Z \in D(F)} \E[\xi^u Z] & = \max_{u \in [0,\frac{1}{3}]} \min_{ Z \in D(F)} \frac{1}{3} \begin{pmatrix} 3u \\ 3-9u \\ 6u \end{pmatrix}Z \\ 
	& = \max_{u \in [0,\frac{1}{3}]} \min \left(\begin{pmatrix} u \\ 1-3u \\ 2u \end{pmatrix}\begin{pmatrix} z \\ x \\ y\end{pmatrix}, 
	\begin{pmatrix} u \\ 1-3u \\ 2u \end{pmatrix}\begin{pmatrix} z \\ y \\ x\end{pmatrix},
	\begin{pmatrix} u \\ 1-3u \\ 2u \end{pmatrix}\begin{pmatrix} y \\ z \\ x\end{pmatrix}
	\right)\\ 
	& = \max_{u \in [0,\frac{1}{3}]} \min \Bigl((z+2y-3x)u+x\ ,\ (z-3y+2x)u+y\ ,\ (y-3z+2x)u+z 
	\Bigr).
\end{align*}
We first study as a function of $u$, which of the three terms is minimum.
\begin{align*}
	\max_{\xi^u} \min_{ Z \in D(F)} \E[\xi^u Z] 
	&= \max \Biggl(\max_{u \in [0,\frac{1}{5}]} (z+2y-3x)u+x, \max_{u \in [\frac{1}{5}, \frac{1}{4}]}  (z-3y+2x)u+y,\Biggr. \\
	&\phantom{=====} \Biggl. \max_{u \in [\frac{1}{4},\frac{1}{3}]} (y-3z+2x)u+z \Biggr).
\end{align*}
We then observe that $z+2y-3x>0$ and $y+2x-3z<0$. Thus,
\begin{align*}
	\max_{\xi^u} \min_{ Z \in D(F)} \E[\xi^u Z] 
	&= \max \Biggl(\frac{z+2y-3x}{5}+x, \max_{u \in [\frac{1}{5}, \frac{1}{4}]}  (z-3y+2x)u+y, \frac{y-3z+2x}{4}+z \Biggr).
\end{align*}
After some simple further calculations, 
\begin{align*}
	\max_{\xi^u} \min_{ Z \in D(F)} \E[\xi^u Z] 
	&= \left\{\begin{array}{cl}\frac{z+2y+2x}{5}&\text{if}\ 2x-3y+z\leq 0,\\
	\\
	\frac{z+y+2x}{4}&\text{if}\ 2x-3y+z\geq 0.\\
	\end{array}   \right.
\end{align*}
The results are summarized in the first row of Table~\ref{tab:optimizer} in which we list the optimizers $(Z^*,\xi^*)$ in each of the cases.

\subsection{Minimax Problem \eqref{dist0}}

For the minimax problem we get by doing all cases explicitly 
\begin{align*}
	& \phantom{=:}\min_{ Z \in D(F)} \max_{\xi^u} \E[\xi^u Z] \\& = \min \Biggl(\max_{u \in [0,\frac{1}{3}]} \begin{pmatrix} u \\ 1-3u \\ 2u \end{pmatrix} \begin{pmatrix} x \\ y \\ z\end{pmatrix}, \max_{u \in [0,1/3]} \begin{pmatrix} u \\ 1-3u \\ 2u \end{pmatrix} \begin{pmatrix} x \\ z \\ y\end{pmatrix}, \max_{u \in [0,1/3]} \begin{pmatrix} u \\ 1-3u \\ 2u \end{pmatrix} \begin{pmatrix} y \\ x \\ z\end{pmatrix}, \Biggr.\\
	& \phantom{====} \Biggl. \max_{u \in [0,\frac{1}{3}]} \begin{pmatrix} u \\ 1-3u \\ 2u \end{pmatrix} \begin{pmatrix} y \\ z \\ x\end{pmatrix}, \max_{u \in [0,\frac{1}{3}]} \begin{pmatrix} u \\ 1-3u \\ 2u \end{pmatrix} \begin{pmatrix} z \\ x \\ y\end{pmatrix}, \max_{u \in [0,1/3]} \begin{pmatrix} u \\ 1-3u \\ 2u \end{pmatrix} \begin{pmatrix} z \\ y \\ x\end{pmatrix} \Biggr)\\
	&= \min \Biggl(\max_{u \in [0,\frac{1}{3}]} (x-3y+2z)u+y\, ,\  z\, ,\   \frac{y}{3}+\frac{2z}{3}\, ,\   z\, ,\  \frac{z}{3}+\frac{2y}{3}\, ,\ \max_{u \in [0,\frac{1}{3}]} (z-3y+2x)u+y\Biggr),
\end{align*}
using the fact that $x<y<z$. The optimal solution is then obtained by splitting into cases based on the respective signs of $x-3y+2z$ and $z-3y+2x$. Note that $2z+x>2x+z$ and thus if $x+2z\le 3y$ then $2x+z-3y<0$. The solutions are reported in the fourth row of Table~\ref{tab:optimizer}. 
 
Specifically, observe that the optimal solution to the minimax problem  can give a different value than the maximin problem. Notice that the $Z$ component of the optimal solution can also be part of an optimal solution to the maximin problem, but the corresponding pricing kernels are different.

\subsection{Convexified Minimax Problem \eqref{cvxminimaxA}}

For the convexified version of the minimax problem we have
\begin{align*}
	\min_{\conv(F)} \max_{\xi^u} \E[\xi^u Z] & = \min_{\substack{x \leq a,b \leq z \\ x+y \leq a+b \leq y+z}} \max_{u \in [0,\frac{1}{3}]} \begin{pmatrix} u \\ 1-3u \\ 2u \end{pmatrix} \begin{pmatrix} a \\ b \\ x+y+z-a-b\end{pmatrix} \\
	&= \min_{\substack{x \leq a,b \leq z \\ x+y \leq a+b \leq y+z}} \max_{u \in [0,\frac{1}{3}]} \Bigl( b + \bigl(2(x+y+z)-a-5b\bigr)u\Bigr) \\ &= \min \biggl(\min_{\substack{x \leq a,b \leq z \\ x+y \leq a+b \leq y+z \\ a+5b < 2(x+y+z)}} \Bigl(\frac{2}{3}(x+y+z) - \frac{a}{3} - \frac{2b}{3}\Bigr)\ ,\  \min_{\substack{x \leq a,b \leq z \\ x+y \leq a+b \leq y+z \\ a+5b \geq 2(x+y+z)}} b \biggr).
\end{align*}
We solve two linear programming problems. The solution to the first linear programming problem is a corner solution: the optimum is obtained for $(a^*,b^*)=(y,x)$ and the minimum value is $y/3+2z/3$. Thus
\begin{align*}
	\min_{\conv(F)} \max_{\xi^u} \E[\xi^u Z] & = \min \biggl(\frac{y}{3}+\frac{2z}{3}\ ,\  \min_{\substack{x \leq a,b \leq z \\ x+y \leq a+b \leq y+z \\ a+5b \geq 2(x+y+z)}} b\biggr) .
\end{align*}
The solution to  the second linear programming problem is also a corner solution. We need to solve for the intersection $(a_{1},b_{1})$ between $a+5b=2(x+y+z)$ and $a+b=y+z$. We find that $a_{1}=(3y+3z-2x)/4$ and $b_{1}=(2x+y+z)/4$. We observe that $a_{1}<z$ if and only if $z-3y+2x>0$. Thus if $z-3y+2x>0$ then  the solution to the minimization is attained at $(a_{1},b_{1})$. Otherwise, it is at $(z,\tilde{b}_{1}):=(z,(2x+2y+z)/5)$. The minimum is then equal to $b_{1}$ or $\tilde b_{1}$ and in both cases, it is less than $y/3+2z/3$. Thus
\begin{align*}
	\min_{\conv(F)} \max_{\xi^u} \E[\xi^u Z] & = \left\{ \begin{array}{ll} \frac{2x+2y+z}{5} &\hbox{if}\ z-3y+2x\le 0,\\
	\frac{2x+y+z}{4}\ &\hbox{if}\ z-3y+2x\ge 0.\end{array}\right.
\end{align*}

\subsection{Convexified Maximin Problem \label{convmaximinproof}}

Finally, for the convexified version of the maximin problem 
\begin{align*}
	\max_{\xi^u} \min_{\conv(F)} \E[\xi^u Z] & = \max_{u \in [0,\frac{1}{3}]} \min_{\substack{x \leq a,b \leq z \\ x+y \leq a+b \leq y+z}} \begin{pmatrix} u \\ 1-3u \\ 2u \end{pmatrix} \begin{pmatrix} a \\ b \\ x+y+z-a-b\end{pmatrix} \\
	&=  \max_{u \in [0,\frac{1}{3}]} \min_{\substack{x \leq a,b \leq z \\ x+y \leq a+b \leq y+z}} \Bigl( 2(x+y+z) u - ua + (1-5u)b \Bigr).
\end{align*}
We split into several cases to solve for the minimum for each possible value of $u$:
\begin{itemize}
	\item[(a)] If $u=0$, the minimum is equal to $x$ and attained at $(a,b)=(h,x)$ for $h\in[y,z]$.
	
	\item[(b)] If $u\in(0,1/5)$, the minimum is equal to $x+u(-3x+2y+z)$ and attained at $(a,b)=(z,x)$. Maximizing over $u$ gives 		$(2x+2y+z)/5$ achieved at $u = \frac{1}{5}$.
	
	\item[(c)] If $u=1/5$, the minimum is equal to $(2x+2y+z)/5$ and attained at $(a,b)=(z,h)$ for $h\in[x,y]$.
	
	\item[(d)] If $u\in(1/5,1/4)$, the minimum is equal  to $y+u(2x-3y+z)$ and attained at $(a,b)=(z,y)$. In the case that $2x-3y+z>0$, the maximum is $(2x+y+z)/4$ and is attained at $u = 1/4$. In the case that $2x-3y+z<0$, the maximum is $(2x+2y+z)/5$ and is attained at $u = 1/5$. Finally, if $2x-3y+z=0$, the maximum $y$ and is attained at all $u \in \bigl(1/5, 1/4\bigr)$.
	
	\item[(e)] If $u=1/4$, the minimum is equal to $(2x+y+z)/4$ and attained at $(a,b)=(\lambda y + (1-\lambda)z, \lambda z + (1-\lambda)y)$ for $\lambda\in[0,1]$.
	
	\item[(f)] If $u\in(1/4, 1/3)$, the minimum is equal to  $z+u(2x+y-3z)$ and attained at $(a,b)=(y,z)$. Maximizing over $u$ gives $(2x+y+z)/4$ achieved at $u = 1/4$.
\end{itemize}

Now, in the case that $2x+z-3y>0$ the optimum is achieved in case (e) yielding the solution set 
\[
	\biggl\{\biggl((z-t,y+t,x), \Bigl(\frac{3}{4},\frac{3}{4},\frac{3}{2}\Bigr)\biggr) \, : \, t \in [0, z-y]\biggr\}.
\]
If $2x+z-3y<0$ the optimum appears in case (c) and is 
\[
	\biggl\{\biggl((z,y-t,x+t), \Bigl(\frac{3}{5},\frac{3}{5},\frac{6}{5}\Bigr)\biggr) \, : \, t \in [0, y-x]\biggr\}.
\]
Finally, if $2x+z-3y=0$ the optimum is achieved in (c), (d) and (e) and given by 
\begin{align*}
	& \biggl\{\biggl((z-t,y+t,x), \Bigl(\frac{3}{4},\frac{3}{4},\frac{3}{2}\Bigr)\biggr) \, : \, t \in [0, z-y]\biggr\},  \qquad \biggl\{\bigl((z,y,x), \xi^u\bigr) \, : \, u \in \Bigl(\frac{1}{5},\frac{1}{4}\Bigr)\biggr\}, \\ 
	& \biggl\{\biggl((z,y-t,x+t), \Bigl(\frac{3}{5},\frac{3}{5},\frac{6}{5}\Bigr)\biggr) \, : \, t \in [0, y-x]\biggr\}.
\end{align*}

We give an illustration of the optimizers in Figure \ref{fig2} hereafter. 
In particular, we find that optimizers coincide only in the perfectly cost-efficient case, and even then some of the optimizers to each problem are part of this shared solution.

 A comparative analysis of the optima and optimizers of the four problems is displayed in Tables~\ref{tab:optimum} and~\ref{tab:optimizer}.

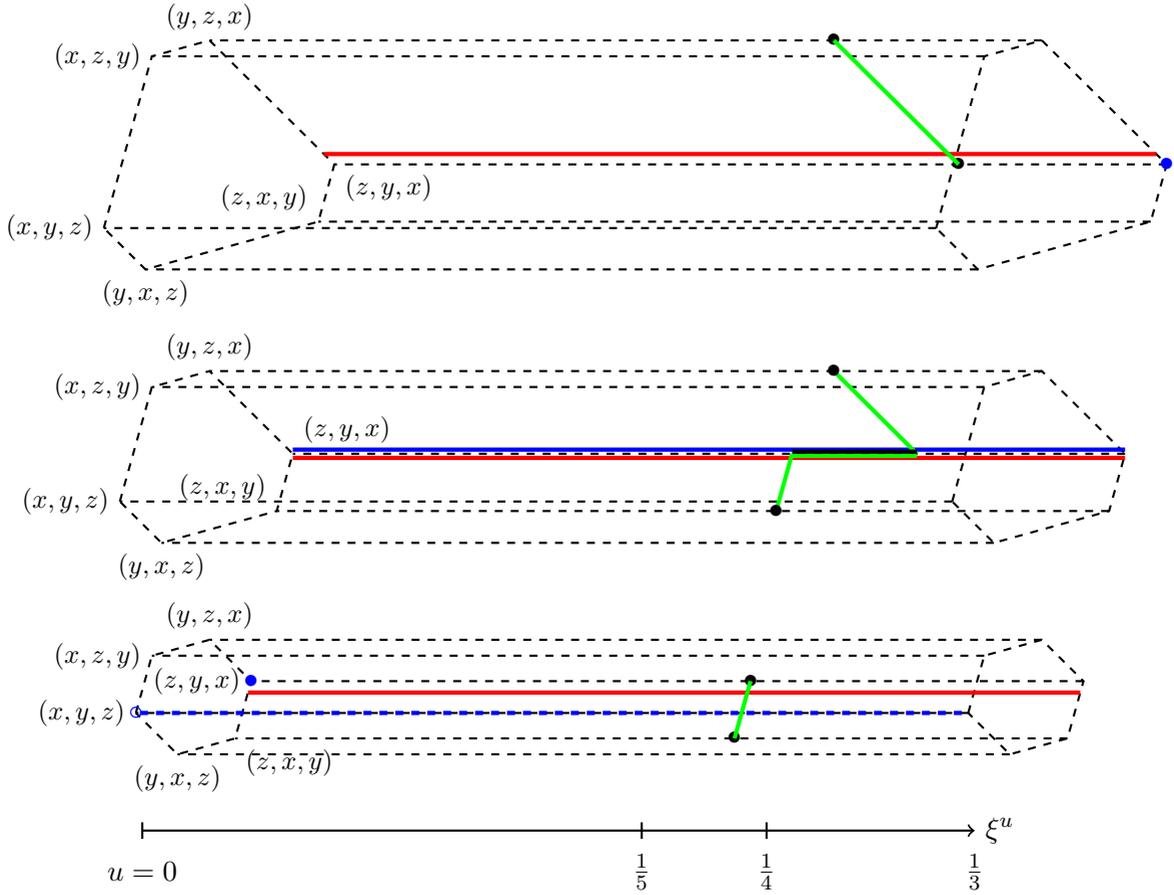
\begin{figure}[tbh!]
	\centering
		\begin{tikzpicture}[scale=0.55]
			\draw[thick,dashed] (1,2,4) -- (1,4,2) -- (2,4,1) --  (4,2,1)  -- (4,1,2) --(2,1,4) -- cycle;
			\draw[thick,dashed] (21,2,4) -- (21,4,2) -- (22,4,1) -- (24,2,1) -- (24,1,2) -- (22,1,4) -- cycle;
			\draw[thick,dashed] (21,2,4) -- (1,2,4) node[anchor= east]{\small $(x,y,z)$};
			\draw[thick,dashed] (22,1,4) -- (2,1,4) node[anchor= north]{\small $(y,x,z)$};
			\draw[thick,dashed] (22,4,1) -- (2,4,1) node[anchor= south]{\small $(y,z,x)$};
			\draw[thick,dashed] (24,2,1) -- (4,2,1) node[anchor= south west]{\small $(z,y,x)$};
			\draw[thick,dashed] (24,1,2) -- (4,1,2) node[anchor= south east]{\small $(z,x,y)$};
			\draw[thick,dashed] (21,4,2) -- (1,4,2) node[anchor= east]{\small $(x,z,y)$};
			\draw[ultra thick,red] (24,1.9,1) -- (4,1.9,1);
			\draw[ultra thick,blue] (24,2.1,1) -- (4,2.1,1);
			\draw[ultra thick,green] (16,1.95,1) -- (19,1.95,1);
			\draw[ultra thick,green] (16,2,1) -- (16,1,2);
			\draw[ultra thick,green] (19,2,1) -- (17,4,1);
			\draw[ultra thick,black] (16,2.05,1) -- (19,2.05,1) node at (17,4,1) {$\bullet$} node at (16,1,2) {$\bullet$};
			\begin{scope}[yshift=7cm]
				\draw[thick,dashed] (1,2,5) -- (1,5,2) -- (2,5,1) -- (5,2,1) -- (5,1,2) -- (2,1,5) -- cycle;
				\draw[thick,dashed] (21,2,5) -- (21,5,2) -- (22,5,1) -- (25,2,1) -- (25,1,2) -- (22,1,5) -- cycle;
				\draw[thick,dashed] (21,2,5) -- (1,2,5) node[anchor= east]{\small $(x,y,z)$};
				\draw[thick,dashed] (22,1,5) -- (2,1,5) node[anchor= north]{\small $(y,x,z)$};
				\draw[thick,dashed] (22,5,1) -- (2,5,1) node[anchor= south]{\small $(y,z,x)$};
				\draw[thick,dashed] (25,2,1) -- (5,2,1) node[anchor= north west]{\small $(z,y,x)$};
				\draw[thick,dashed] (25,1,2) -- (5,1,2) node[anchor= south east]{\small $(z,x,y)$};
				\draw[thick,dashed] (21,5,2) -- (1,5,2) node[anchor= east]{\small $(x,z,y)$};
				\node[blue] at (25,2,1) {$\bullet$};
				\draw[ultra thick,red] (24.75,2.25,1) -- (4.75,2.25,1);
				\node[ultra thick,black] at (20,2,1) {$\bullet$} node at (17,5,1) {$\bullet$};
				\draw[ultra thick,green] (17,5,1) -- (20,2,1);
			\end{scope}
			\begin{scope}[yshift=-5.5cm]
				\draw[thick,dashed] (1,2,3) -- (1,3,2) -- (2,3,1) -- (3,2,1) -- (3,1,2) -- (2,1,3) -- cycle;
				\draw[thick,dashed] (21,2,3) -- (21,3,2) -- (22,3,1) -- (23,2,1) -- (23,1,2) -- (22,1,3) -- cycle;
				\draw[thick,dashed] (21,2,3) -- (1,2,3) node[anchor= east]{\small $(x,y,z)$};
				\draw[thick,dashed] (22,1,3) -- (2,1,3) node[anchor= north]{\small $(y,x,z)$};
				\draw[thick,dashed] (22,3,1) -- (2,3,1) node[anchor= south]{\small $(y,z,x)$};
				\draw[thick,dashed] (23,2,1) -- (3,2,1) node[anchor= east]{\small $(z,y,x)$};
				\draw[thick,dashed] (23,1,2) -- (3,1,2) node[anchor= north west]{\small $(z,x,y)$};
				\draw[thick,dashed] (21,3,2) -- (1,3,2) node[anchor= east]{\small $(x,z,y)$};
				\node at (0,-3,0) {$u=0$};
				\node at (12,-3,0) {$\frac{1}{5}$};
				\node at (15,-3,0) {$\frac{1}{4}$};
				\node at (20,-3,0) {$\frac{1}{3}$};
				\draw[ultra thick,red] (23,1.8,1.2) -- (3,1.8,1.2);
				\draw[ultra thick,blue, dashed] (1.1,2,3) -- (21,2,3) node at (1,2,3) {$\mathbf{\circ}$} node at (3,2,1) {$\bullet$};
				\node[ultra thick,black] at (15,1,2) {$\bullet$} node at (15,2,1) {$\bullet$};
				\draw[ultra thick,green] (15,1,2) -- (15,2,1);
				\draw[thick,->] (0,-2,0) -- (20,-2,0) node[anchor= west]{$\xi^u$};
				\draw[thick] (0,-1.8,0) -- (0,-2.2,0);
				\draw[thick] (12,-1.8,0) -- (12,-2.2,0);
				\draw[thick] (15,-1.8,0) -- (15,-2.2,0);
			\end{scope}
		\end{tikzpicture}
	\caption{Illustration of the solutions to the 3-states example in the cases $2x-3y+z>0$ (top, $x=1,\, y=2, \, z=5$) $2x-3y+z=0$ (middle, the cost-efficient case, $x=1,\, y=2, \, z=4$) and $2x-3y+z<0$ (bottom, $x=1,\, y=2, \, z=3$). Convexified Minimax (red), minimax (blue), convexified maximin (green) and maximin (black) problems.}\label{fig2}
\end{figure}

\newpage
\subsection{Direct Proof of Remark~\ref{attainable}: Attainable Perfectly Cost-efficient Payoffs\label{proofattain}}
		
\begin{proof} 
	Consider a distribution $(x,y,z) \in D(F), \, x  < y  < z$. Let us characterize the set of attainable payoffs. Recall that if $(x_{1},x_{2},x_{3})$ be attainable, then $\forall u \in\bigl[0,\frac{1}{3}\bigr],\quad \E\Biggl[\xi^u \Biggl(\begin{array}{c}x_{1}\\x_{2}\\x_{3}\end{array}\Biggr)\Biggr]=cst.$ This is possible if and only if
	\[
		x_{1}-3x_{2}+2x_{3}=0.
	\]

	Using the fact that $x<y<z$, we can discard some of the cases.
	\begin{itemize}
		\item $(x,y,z)$ is attainable if and only if $x-3y+2z=0$.
		\item $(x,z,y)$ is attainable if and only if $x-3z+2y=0$. This is impossible as $x+2y<z$.
		\item $(y,x,z)$ is attainable if and only if $y-3x+2z=0$. This is impossible as $y+2z>3x$.
		\item $(y,z,x)$ is attainable if and only if $y-3z+2x=0$. This is impossible as $2x+y<3z$.
		\item $(z,y,x)$ is attainable if and only if $z-3y+2x=0$.
		\item $(z,x,y)$ is attainable if and only if $z-3x+2y=0$. This is impossible as $z+2y>3x$.
	\end{itemize}

Overall, the only possible payoffs that are attainable with distribution $\{x,y,z\}$ with probability $1/3$ are $(x,y,z)$ and $(z,y,x)$.

Then  if in addition, this payoff writes as $F^{-1}\bigl(1-U_{\xi^u})\bigr)$ (or a randomization of it) for some pricing kernel $\xi^u$ , then $(x,y,z)$ is impossible (as it would imply that  $3u \geq 3-9u \geq 6u$). Thus, the only possible payoff that writes as $F^{-1}\bigl(1-U_{\xi^u}\bigr)$ is $(z,y,x)$ for $u \in [1/5,1/4]$. Specifically $\xi^u$ takes 3 distinct values for $3u<3-9u<6u$, which is equivalent to $u\in(1/5,1/4)$, whereas in the two borderline cases with $u = \frac{1}{5}$ and $u = \frac{1}{4}$ the pricing kernel takes only two values, and the payoff is achieved through randomization.

We thus find that a payoff that is attainable and writes as $F^{-1}\bigl(1-U_{\xi^u})\bigr)$ if and only if it is  perfectly cost-efficient, i.e., $z= 3y-2x$. This is in line with the finding in \cite[Theorem~5.3]{BS24}.
\end{proof}

\subsection{On  \cite[Proposition~4.5]{BS24}
\label{proofprop4.6}}

Here we want to exemplify the idea of KKM argument in the proof of \cite[Proposition~4.5]{BS24} in the 3-state model. We note that
\[
	F^{-1}\bigl(1-\hat{F}_{\xi^u}(\xi^u; \hat{U}_{\xi^u})\bigr) = 
	\left\{\begin{array}{ll} (z,x,y) & \text{ if } u \in \bigl(0, \frac{1}{5}\bigr), \\
		(z,x,y) \text{ and } (z,y,x) \text{ with equal probability} & \text{ if } u =\frac{1}{5},\\
		(z,y,x) & \text{ if } u \in \bigl(\frac{1}{5}, \frac{1}{4}\bigr), \\
		(z,y,x) \text{ and } (y,z,x) \text{ with equal  probability} & \text{ if } u =\frac{1}{4}, \\
		(y, z,x) & \text{ if } u \in \bigl(\frac{1}{4}, \frac{1}{3}\bigr)
	\end{array}\right.
\]
and thus 
\[
	e(s,u) := \E\Bigr[\xi^sF^{-1}\bigl(1-\hat{F}_{\xi^u}(\xi^u; \hat{U}_{\xi^u})\bigr)\Bigr] = 
	\left\{\begin{array}{ll} x + (-3x + 2y +z)s & \text{ if } u \in \bigl(0, \frac{1}{5}\bigr), \\ 
		\frac{x+y}{2} + \bigl(-\frac{x+y}{2}+z\bigr)s  & \text{ if } u =\frac{1}{5},
		\\ y + (2x-3y+z)s & \text{ if } u \in \bigl(\frac{1}{5}, \frac{1}{4}\bigr), \\
		\frac{y+z}{2} +(2x-y-z)s & \text{ if } u =\frac{1}{4}, \\
		z+(2x+y-3z)s & \text{ if } u \in \bigl(\frac{1}{4}, \frac{1}{3}\bigr).
	\end{array}\right.
\]
From this we note that in the case $2x-3y+z >0$ we have
\[
	\mathcal{A}(\xi^s) = \mathcal{B}(\xi^s) = 
	\left\{\begin{array}{ll} \bigl[s,\frac{1}{4}\bigr] & \text{ if } s \in \bigl(0, \frac{1}{4}\bigr), \\
		\bigl\{\frac{1}{4}\bigr\} & \text{ if } s = \frac{1}{4}, \\
		\bigl[\frac{1}{4},s\bigr] & \text{ if } s \in \bigl(\frac{1}{4}, \frac{1}{3}\bigr),
	\end{array}\right.
\]
which is compact and closed in probability and thus
\[
	\bigcap_{s \in \bigl(0, \frac{1}{3}\bigr)} \mathcal{A}(\xi^s) =  \biggl\{\frac{1}{4}\biggr\}.
\]
Similarly in the case $2x-3y+z <0$ we have
\[
	\mathcal{A}(\xi^s) = \mathcal{B}(\xi^s) = \left\{\begin{array}{ll} \bigl[s,\frac{1}{5}\bigr] & \text{ if } s \in \bigl(0, \frac{1}{5}\bigr) \\ \bigl\{\frac{1}{5}\bigr\} & \text{ if } s = \frac{1}{5} \\\bigl[\frac{1}{5},s\bigr] & \text{ if } s \in \bigl(\frac{1}{5}, \frac{1}{3}\bigr) \end{array}\right.
\]
and thus
\[
	\bigcap_{s \in \bigl(0, \frac{1}{3}\bigr)} \mathcal{A}(\xi^s) =  \biggl\{\frac{1}{5}\biggr\} .
\]
Finally, if $2x-3y+z = 0$ we have
\[
	\mathcal{A}(\xi^s) = \mathcal{B}(\xi^s) =
	\left\{\begin{array}{ll} \bigl[s,\frac{1}{4}\bigr] & \text{ if } s \in \bigl(0, \frac{1}{5}\bigr),\\
		\bigl[\frac{1}{5},\frac{1}{4}\bigr] & \text{ if } s \in \bigl[\frac{1}{5},\frac{1}{4}\bigr], \\
		\bigl[\frac{1}{5},s\bigr] & \text{ if } s \in \bigl(\frac{1}{4}, \frac{1}{3}\bigr)
	\end{array}\right.
\]
and thus
\[
	\bigcap_{s \in \bigl(0, \frac{1}{3}\bigr)} \mathcal{A}(\xi^s) =  \biggl[\frac{1}{5},\frac{1}{4}\biggr].
\]
An illustration for this argument is given in Figure~\ref{fig:KKM}.

\begin{figure}[htb]
	\centering
		\begin{tikzpicture}[scale=2.2]
		\draw[gray!80, ->] (-0.5cm,0cm) -- (3.5cm,0cm) node at (3.55,-0.1) {$s$};  
		\draw[gray!80, ->] (0cm,-0.5cm) -- (0cm,3.4cm) node at (-0.3,3.4) {$e(s,u)$};  
		\draw[gray!60, dotted] (1.66cm,-0.5cm) -- (1.66cm,3.4cm) node at (1.66,-0.4) {$\frac{1}{6}$};  
		\draw[gray!60, dotted] (2cm,-0.5cm) -- (2cm,3.4cm) node at (2,-0.4) {$\frac{1}{5}$};  
		\draw[gray!60, dotted] (2.5cm,-0.5cm) -- (2.5cm,3.4cm) node at (2.5,-0.4) {$\frac{1}{4}$};  
		\draw[gray!60, dotted] (2.86cm,-0.5cm) -- (2.86cm,3.4cm) node at (2.86,-0.4) {$\frac{2}{7}$};  
		\draw[gray!60, dotted] (3.33cm,-0.5cm) -- (3.33cm,3.4cm) node at (3.33,-0.4) {$\frac{1}{3}$};  \
		\draw[blue!80, dashed] (0cm,1cm) -- (3.33cm,2.33cm) node at (3.8,2.33) {$u \in \bigl(0,\frac{1}{5}\bigr)$}; 
		\draw[olive!80, dotted] (0cm,1.5cm) -- (3.33cm,2cm) node at (3.8,2) {$u = \frac{1}{5}$}; 
		\draw[green!80, dashed] (0cm,2cm) -- (3.33cm,1.67cm) node at (3.8,1.67) {$u \in \bigl(\frac{1}{5},\frac{1}{4}\bigr)$}; 
		\draw[orange!80, dotted] (0cm,2.5cm) -- (3.33cm,1.5cm) node at (3.8,1.5) {$u = \frac{1}{4}$}; 
		\draw[yellow!80, dashed] (0cm,3cm) -- (3.33cm,1.33cm) node at (3.8,1.33) {$u \in \bigl(\frac{1}{4},\frac{1}{3}\bigr)$}; 
		\draw[red, thick] (0cm,1cm) -- (2cm, 1.8cm) -- (2.5cm,1.75cm) -- (3.33cm,1.33cm) node at (0.8,1.1) {$e(s,s)$}; 
		\draw[red, thick] (1.66cm,0cm) -- (2cm, 0cm) node at (1.8,-0.2) {$\mathcal{B}(\xi^\frac{1}{6})$}; 
		\draw[red, thick] (1.7cm,0.04cm) -- (1.66cm, 0.04cm) --  (1.66cm, -0.04cm) -- (1.7cm,- 0.04cm); 
		\draw[red, thick] (1.96cm,0.04cm) -- (1.995cm, 0.04cm) --  (1.995cm, -0.04cm) -- (1.96cm,- 0.04cm); 
		\draw[blue, thick] (2cm,0cm) -- (2.86cm, 0cm) node at (2.2,-0.2) {$\mathcal{B}(\xi^\frac{2}{7})$}; 
		\draw[blue, thick] (2.82cm,0.04cm) -- (2.86cm, 0.04cm) --  (2.86cm, -0.04cm) -- (2.82cm,- 0.04cm); 
		\draw[blue, thick] (2.04cm,0.04cm) -- (2.005cm, 0.04cm) --  (2.005cm, -0.04cm) -- (2.04cm,- 0.04cm); 
		\end{tikzpicture}
	\caption{Illustration of the KKM argument in  \cite[Proposition~4.5]{BS24} for the 3-state model with $x=1$, $y=2$ and $z=3$.}\label{fig:KKM}
\end{figure}
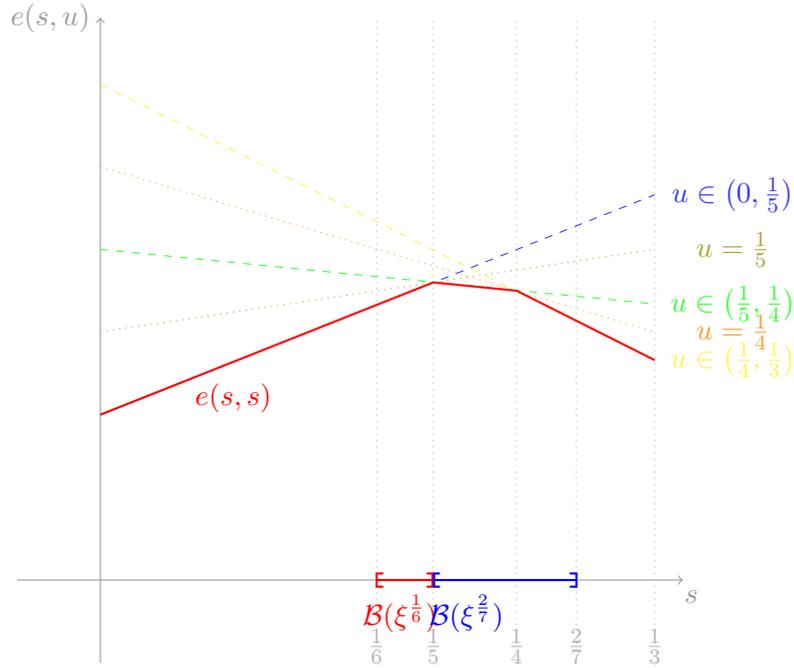

\section{Detailed calculations for Example~\ref{coexa}\label{Appcoex}}

We provide here the detailed calculations to prove that problems \eqref{dist0} and \eqref{dist-problem} do not share the same solutions in the market model presented in Example~\ref{coexa}. We have the following explicit expressions		

\begin{align*}
	F_{S_{T}}(x) & = p\Phi\Biggl(\frac{\ln(x/S_{0})-(\mu-0.5\sigma_{H}^2)T}{\sigma_{H}\sqrt{T}}\Biggr)+ (1-p)\Phi\Biggl(\frac{\ln(x/S_{0})-(\mu-0.5\sigma_{L}^2)T}{\sigma_{L}\sqrt{T}}\Biggr),\\
	F_{\xi_{T}^q}(x) & =p\Phi\Biggl(\frac{\ln(x)+rT-\ln(q/p)+\theta_{H}^2T/2}{\theta_{H}\sqrt{T}}\Biggr)\\
	&\phantom{====}+ (1-p)\Phi\Biggl(\frac{\ln(x)+rT-\ln((1-q)/(1-p))+\theta_{L}^2T/2}{\theta_{L}\sqrt{T}}\Biggr).
\end{align*}

From \cite{bernard2015quantile}, we have an explicit expression of the quantile function for $S_{T}$ and $\xi_{T}^q.$ Define $X_{H}$ and $X_{L}$ such that $S_{T}=X_{H}\ind_{\{\sigma=\sigma_{H}\}}+X_{L}\ind_{\{\sigma=\sigma_{L}\}}.$ Let 
\[
	\alpha_u^*:=\inf\Biggl\{\alpha\in\Biggl(\max\biggl(\frac{u-(1-p)}{p}, 0\biggr), \, \min\Bigl(1, \frac{u}{p}\Bigr)\Biggr) \, : \,  F_{X_{H}}^{-1}(\alpha)\ge F_{X_{L}}^{-1}\biggl(\frac{u-\alpha p}{1-p}\biggr)\Biggr\}
\]

In general $F_{S_{T}}^{-1}(u)=\max\bigl(F_{X_{H}}^{-1}(\alpha^*_{u}),F_{X_{L}}^{-1}\bigl(\frac{u-\alpha_{u}^* p}{1-p}\bigr)\bigr).$ In our case, the distributions involved are continuous, thus these two quantities are equal and in addition for all $u\in(0,1)$, $\alpha_{u}^*\in(0,1)$ with \[
	F_{S_{T}}^{-1}(u) = F_{X_{H}}^{-1}(\alpha_{u}^*) = S_{0}\exp\biggl(\sigma_{H}\sqrt{T}\Phi^{-1}\bigl(\alpha_{u}^*\bigr)- \frac{\sigma_{H}^2T}{2}+\mu T\biggr),
\]
where $\alpha^*_{u}$ is the unique root of $F_{X_{H}}^{-1}(\alpha_{u}^*) = F_{X_{L}}^{-1}\Bigl(\frac{u-\alpha_{u}^* p}{1-p}\Bigr)$, i.e.,
\[
	\sigma_{H}\sqrt{T}\Phi^{-1}(\alpha^*_{u})-\frac{\sigma_{H}^2T}{2}=\sigma_{L}\sqrt{T}\Phi^{-1}\Bigl(\frac{u-\alpha_{u}^* p}{1-p}\Bigr)-\frac{\sigma_{L}^2T}{2}.
\]
Similarly, $F_{\xi^q_{T}}^{-1}(u)=\frac{q}{p}\exp\bigl(\theta_{H}\sqrt{T}\Phi^{-1}\bigl(\gamma_{u}^*\bigr)-\theta_{H}^2T/2-rT\bigr)$ where $\gamma^*_{u}$ is the unique root of 
\[
	\frac{q}{p}\exp\biggl(\theta_{H}\sqrt{T}\Phi^{-1}\bigl(\gamma_{u}^* \bigr)-\frac{\theta_{H}^2T}{2}\biggr)=\frac{1-q}{1-p}\exp\biggl(\theta_{L}\sqrt{T}\Phi^{-1}\Bigl(\frac{u-\gamma_{u}^* p}{1-p}\Bigr)-\frac{\theta_{L}^2T}{2}\biggr).
\]
We can prove that
\begin{equation}\label{interm}
	\sup_{\xi^q \in \Xi}  \E\bigl[ \xi^q F_S^{-1}\bigl(1-F_{\xi^q}(\xi^q)\bigr)\bigr] =\E\bigl[ \xi^* F_S^{-1}\bigl(1-F_{\xi^*}(\xi^*)\bigr)\bigr]
\end{equation}
where $\xi^*=\xi^{q_{0}}$ for $q_{0}\in(0,1)$. To do so, observe that
\[
	g(q):=\E\bigl[ \xi^q F_S^{-1}\bigl(1-F_{\xi^q}(\xi^q)\bigr)\bigr]=\E\bigl[ F_{\xi^q}^{-1}(U) F_S^{-1}(1-U)\bigr]
\]
for some uniformly distributed r.v. $U$, so that
\[
	g^\prime(q)=\frac{d}{dq}\int_{0}^1F_{\xi^q}^{-1}(u) F_S^{-1}(1-u)du=\int_{0}^1\frac{d}{dq}\Bigl(F_{\xi^q}^{-1}(u)\Bigr) F_S^{-1}(1-u)du,
\]
which can be written as
\[
	g^\prime(q)=\int_{0}^1\biggl( \frac{1}{p}+\frac{q}{p}\theta_{H}\sqrt{T}\Bigl(\Phi^{-1}\Bigr)^\prime\bigl(\gamma_{u}^* \bigr)\frac{d}{dq}(\gamma_{u}^*)\biggr)e^{\theta_{H}\sqrt{T}\Phi^{-1}\bigl(\gamma_{u}^*\bigr)-\frac{\theta_{H}^2T}{2}-rT} F_S^{-1}(1-u)du,
\]
we have that $g^\prime(0)>0$ and $g^\prime(1)<0$ ensuring that the optimal value for $q$ is obtained in the interior of $(0,1).$

Furthermore, observe that if $(\bar\xi,\bar X)$ is a solution to
\begin{equation}\label{numeric}
	\inf_{\substack{Z \in \mathcal{X} \\ Z \in D(F_S)}}  \sup_{\xi \in \Xi} \E\bigl[ \xi Z\bigr],
\end{equation}
then $\bar\xi$ must be equal to $\xi^{0}$ or $\xi^{1},$ given that 
\[
	\E\bigl[ \xi^q Z\bigr]=\frac{q}{p} \E\bigl[Z Y_{H}\ind_{\{\sigma=\sigma_{H}\}}\bigr]  +\frac{1-q}{1-p} \E\bigl[Z Y_{L}\ind_{\{\sigma=\sigma_{L}\}}\bigr], 
\]
where $Y_{H}$ and $Y_{L}$ do not depend on $q$. This expression is linear in the variable $q$ and thus its  optimum is obtained at $q=0$ or $q=1$. Therefore $\bar\xi\neq\xi^*$. However, if there is an equality in \eqref{CEFF} then
\begin{equation}\label{numeric2}
	\sup_{\xi \in \Xi} \, \inf_{\substack{Z \in \mathcal{X} \\ Z \in D(F_S)}} \E\bigl[ \xi Z\bigr]=  \E\bigl[ \xi^* Z^*\bigr]=\E\bigl[ \bar\xi \bar Z\bigr],
\end{equation}
thus $(\bar\xi,\bar X)$ is also solution and by uniqueness of solution to \eqref{dist-problem} (\cite[Proposition~4.4]{BS24}), we must have $\bar\xi=\xi^*,$ which contradicts the above results. We are able to conclude that the cost of the cheapest attainable payoff to achieve the distribution of final wealth $F_{S}$ is strictly larger than the left side of the inequality in  \cite[Proposition~4.2]{BS24}.

\end{document}